\documentclass[10pt, conference, compsocconf]{IEEEtran}  
\usepackage{tikz}
\usepackage{textcomp}
\usepackage{hyperref}
\usepackage{lipsum}

\newcommand\copyrighttext{%
  \footnotesize \textcopyright 2018 IEEE. Personal use of this material is permitted.
  Permission from IEEE must be obtained for all other uses, in any current or future
  media, including reprinting/republishing this material for advertising or promotional
  purposes, creating new collective works, for resale or redistribution to servers or
  lists, or reuse of any copyrighted component of this work in other works.
  DOI: \href{< DOI: 10.1109/Cybermatics_2018.2018.00184 }{ 10.1109/Cybermatics\_2018.2018.00184 }}
\newcommand\copyrightnotice{%
\begin{tikzpicture}[remember picture,overlay]
\node[anchor=south,yshift=10pt] at (current page.south) {\fbox{\parbox{\dimexpr\textwidth-\fboxsep-\fboxrule\relax}{\copyrighttext}}};
\end{tikzpicture}%
}

\usepackage{todonotes}
\usepackage{url}
\usepackage{caption} 

\usepackage[linesnumbered,algoruled,boxed,lined]{algorithm2e}

\usepackage{amsmath}
\usepackage{amsfonts}
\usepackage{amssymb}
\usepackage{amsthm}

\usepackage{mathtools}

\newtheorem{theorem}{Theorem}
\newtheorem{lemma}[theorem]{Lemma}

\newtheorem{definition}{Definition}

\begin{document}

\sloppy

\title{Window Based BFT Blockchain Consensus}

\author{
\IEEEauthorblockN{Mohammad M. Jalalzai}
\IEEEauthorblockA{
Computer Science and Engineering Division\\
Louisiana State University\\
Baton Rouge, Louisiana, USA\\
Email: \url{mjalal7@lsu.edu}}
\and 
\IEEEauthorblockN{Costas Busch}
\IEEEauthorblockA{
Computer Science and Engineering Division\\
Louisiana State University\\
Baton Rouge, Louisiana, USA\\
Email: \url{busch@csc.lsu.edu}}
}

\date{}

\maketitle
\copyrightnotice

\begin{abstract}

Proof of Work (PoW) and Byzantine Fault Tolerant (BFT) are the two main classes of consensus protocols that are used in the blockchain consensus layer. PoW is highly scalable but very slow with performance of about 7 transactions/second. BFT-based protocols are highly efficient for small networks, but their scalability is limited to only tens of nodes. One of the main reasons for the BFT limitation is the quadratic $O(n^2)$ communication complexity of BFT-based protocols for $n$ nodes, which requires $n \times n$ broadcasting. 
In this paper, we present the {\em Musch} protocol which is BFT-based and provides communication complexity $O(f n + n)$ for $f$ failures and $n$ nodes, where $f < n/3$, without compromising the latency.
Hence, the performance adjusts to $f$ such that for constant $f$ the communication complexity is linear. Musch achieves this by introducing the notion of exponentially increasing windows of nodes to which complains are reported, instead of broadcasting to all the nodes.
To our knowledge, this is the first BFT-based blockchain protocol which efficiently addresses simultaneously the issues of communication complexity and latency under the presence of failures. 
\end{abstract}






\section{Introduction}
 
Consensus is used to agree on a new block to be appended to the chain by the nodes in the network. 
A blockchain is compromised of two main components:a cryptographic engine and a consensus engine. The main performance and scalability bottleneck of a blockchain also lies in these components. Here we only focus on improving consensus component of blockchains.
As already mentioned, PoW-based protocols are highly scalable.
In Bitcoin~\cite{Nakamoto_bitcoin:a}, which is one of the most successful implementation of blockchain technology, typically the number of nodes (replicas) are usually large in the range of  thousands~\cite{Nakamoto_bitcoin:a,DBLP:conf/ifip114/Vukolic15}. 
 
PoW involves the calculation of a number based on the hash value of a block adjusted by a difficulty level.
%
%
Solving this cryptographic puzzle  by nodes (miners)  limits the rate of the block generation as solving the puzzle is CPU intensive. 
Bitcoin uses PoW but the number of transactions per second can reach up to just 7 transactions per second~\cite{DBLP:conf/ifip114/Vukolic15}. The block generation rate is  approximately 10 minutes~\cite{Nakamoto_bitcoin:a}.  Additionally, the power utilized by Bitcoin mining in 2014 was between 0.1-10 GW and was comparable to Ireland's electricity consumption at that time
\cite{6912770}.  
Different solutions were proposed, for example, Ethereum \cite{Ethereum-EIP-150} uses faster PoW, BitcoinNG \cite{194906} uses two types of blocks, namely, key blocks 
and micro-blocks,
and has achieved $10 \times$ more throughput in comparison with Bitcoin. But all these solutions fall well short of matching the throughput offered by leading credit-card companies ($2000$ on average and $10000$ maximum transactions per second).

On other the hand, BFT-based \cite{Lamport:1984:UTI:2993.2994}  protocols guarantee consensus in the presence of malicious (Byzantine) nodes, which can fail in arbitrary ways including crashes, software bugs and even coordinated malicious attacks.    Typically, BFT-based algorithms execute in epochs, where in each epoch the correct (non-malicious) nodes achieve agreement for a set of proposed transactions.
In each epoch there is a {\em primary} node that helps to reach agreement.
The consensus is achieved during each epoch and an entry or a set of entries are added to the log.
In case the primary is found to be Byzantine,a  view change (select new primary) takes effect to provide liveness. These protocols have shown the ability to achieve throughput of tens of thousand transactions per second~\cite{Kotla:2008:ZSB:1400214.1400236,Guerraoui:2010:NBP:1755913.1755950}. However, their scalability has been tested with a very small number of nodes $n$, usually 10 to 20 nodes,
due to the requirement for $n \times n$ broadcast \cite{DBLP:conf/ifip114/Vukolic15}, that is, they have quadratic communication complexity. 
 

To address the scalability issues in BFT  protocols, we introduce the {\em Musch} blockchain protocol. Musch is BFT-based and achieves $O(fn + n)$ communication complexity in an epoch, where $f$ is the actual number of Byzantine nodes ($f < n/3$). For small (i.e. constant) $f$ the communication complexity is linear,
and hence, Musch has scalable performance.
Musch does not need to know the actual value of $f$ 
since it automatically adjusts to the actual number of nodes 
that exhibit faulty behavior in each epoch.
At the same time, the latency is comparable with other efficient BFT-based protocols~\cite{Kotla:2008:ZSB:1400214.1400236,Castro:1999:PBF:296806.296824}.


The performance of our algorithm is based on a novel mechanism of communication with a set of {\em window nodes}. Nodes reach sliding windows moving over node IDs to recover from faults during consensus. If a replica does not receive expected messages from the primary, 
it complains to the window nodes from which it recovers updates (see Fig.~\ref{Figure:window}). Initially, the window consists of only one node. If the complainer replica doesn't receive a valid response from any of the the window nodes, it considers the next window of double size to which it sends the complaint. The last window size is no more than $2f$ which guarantees to have a correct node within the last window. This gives $O(fn + n)$ communication complexity. In this way, Musch avoids $n \times n$ broadcasts while guaranteeing consistency. 


\begin{table}
\scriptsize{
\centering
\begin{tabular}{|l|c|c|c|c|}
\hline
& PBFT     & FastBFT   & Aliph   & Musch   \\ \hline
Total Replicas $n$  & $3f'+1$    & $2f'+1$       & $3f'+1$    & $3f'+1$ \\ \hline
Critical Path  & 4        & $3+\log(f'+1)$  & $3f'+2$    & 4 \\ \hline
\begin{tabular}[c]{@{}l@{}}Communication \\Complexity 
\end{tabular}    
& $O(n^2)$ & $O(n^2)$    & $O(n^2)$ & $O(fn + n)$ \\ \hline
\end{tabular}
}
\caption{Characteristics of state of Art BFT protocols. The actual faulty nodes is $f$, while $f'$ is an upper bound, $f \leq f'$.} 
\label{comparisontable}

\end{table}

Table \ref{comparisontable} compares Musch with other state of art BFT-based protocols
such as PBFT \cite{Castro:1999:PBF:296806.296824}, FastBFT \cite{DBLP:journals/corr/LiuLKA16a}, and Aliph \cite{Guerraoui:2010:NBP:1755913.1755950}. 
We compared the communication complexity measured as number of total exchanged messages during an epoch. Our algorithm's performance depends solely on $fn$, while the other algorithms have quadratic communication complexity.
Hence, our algorithm has an advantage when $f$ is asymptotically smaller than $n$,
resulting in less than quadratic communication complexity. 
When $f$ is a constant our algorithm is optimal.
Additionally, we also compared the critical path length, as the number of one-way message latency it takes for a client request to be processed and the response is received by the client.
Note that the total number of nodes is $n = 3f'+1$, where $f'$ is a conservative upper bound on the number of faulty nodes. The actual faults are bounded by $f \leq f'$. Our algorithm does not need to know $f$.

SBFT \cite{DBLP:journals/corr/abs-1804-01626} has also tried to address the issue of scalability and have tested their protocol with $100$ replicas, while achieving $10 \times$ better performance than Ethereum. In normal mode when $f=0$, SBFT's message complexity will be $O(nc)$ ($c$ is the number of collectors) as compared to Musch's $O(n)$. 
The actual value of Byzantine nodes ($f$) has to be known (to choose correct $c$ such that $c\geq f$) for the system to avoid the fall-back protocol. But in practice it is  impossible to know the actual value of $f$(avoiding fallback is not possible for small $c$). Fall-back mode executes efficient PBFT with $O(n^2)$ complexity. This $O(n^2)$ complexity causes additional latency and performance degradation in SBFT.
\section*{Paper Outline}

We continue with this paper as follows.
In Section \ref{section:model},
we give the model of the distributed system.
In Section \ref{section:protocol},
we present our algorithm.
Protocol checkpoints are presented in Section \ref{section:Checkpoints}. 
We give the correctness analysis in Section \ref{section:correctness},
and the communication complexity bound analysis in Section \ref{section:complexity}.

\section{System Model}
\label{section:model}
Like other BFT-based state machine replication protocols Musch also assumes an adversarial failure model. Under this model, servers and even clients may deviate from their normal behavior in arbitrary ways, which includes hardware failures, software bugs, or malicious intent. Our protocol can tolerate up to $f'$ number of Byzantine replicas where the total number of replicas in the network $n = 3f'+1$. Replica ID is an integer from the replica set $\{1, \ldots, n\}$ that identifies each  replica. The actual number of Byzantine replicas in the network is denoted by $f$, and at any moment during execution $0 \leq f \leq f'$. 
If $f=0$ then the execution is fault-free.
However, $f$ may not be known. Our algorithm's communication complexity adapts to any value of $f$.

\begin{figure}
\begin{center}
\includegraphics[width=9cm]{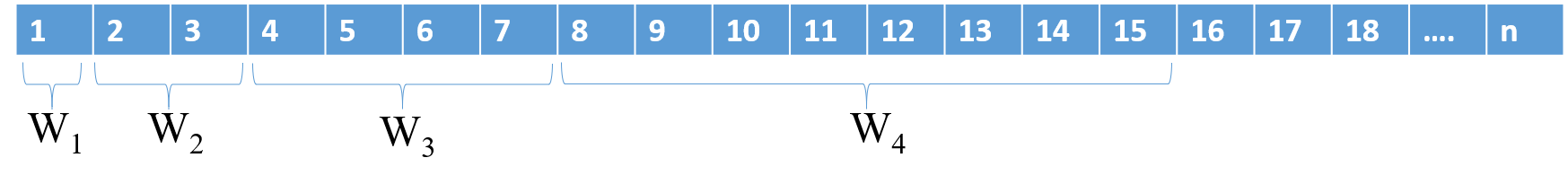}
\end{center}
\caption{Windows of nodes}
\label{Figure:window}
\end{figure}

\section{Protocol}
\label{section:protocol}

Our proposed protocol uses echo broadcast~\cite{Reiter:1994:SAP:191177.191194}, where the primary proposes a block of transactions, and replicas respond by sending back signed hashes of the block.
We assume strong adversarial coordinated attacks by various malicious replicas. However, replicas will not be able to break collision resistant hashes, encryption, or signatures. We assume that all messages sent by replicas and the primary are signed.  For example if primary  $p$ proposes a block of transactions  $\langle B \rangle_p$ to the replica $i$, we assume that it has been signed by primary $p$. Any unsigned message will be discarded.  To avoid repetition of message and signatures, Musch also uses signature aggregation \cite{Boneh:2003} to use a single collective signature instead of appending all replica signatures, to keep signature size constant.
As the primary $p$ receives message $M_i$ with their respective signatures $\sigma_i$ from each replica $i$, the primary then uses these received signatures to generate an aggregated signature $\sigma$. The aggregated signature can be verified by replicas given the messages $M_1,M_2,\ldots, M_y$ where $y \leq n$, the aggregated signature $\sigma$, and public keys $PK_1,PK_2,\ldots,PK_y$. Like other BFT-based protocols \cite{Castro:1999:PBF:296806.296824,Lamport:1982:BGP:357172.357176,Luu:2016:SSP:2976749.2978389} each replica $i$ knows the public keys of other replicas in the network. In Section \ref{Recovery mode} we explain how to use the IDs to define the windows that we use in the algorithm. 

It is not possible to ensure the safety and liveness of consensus algorithms in asynchronous systems where even a single replica can crash fail \cite{Fischer:1985:IDC:3149.214121}.  Musch's safety holds in asynchronous environments. But to circumvent this impossibility for liveness, Musch assumes partial synchrony~\cite{Dwork:1988:CPP:42282.42283}.  This partial synchrony is achieved by using  arbitrarily large unknown but fixed worst case global stabilization delays. 

During normal operation, Musch guarantees that at least $2f'+1$ replicas in each epoch are consistent (out of the $n = 3f'+1$).
Let $T$ be the maximum round-trip message delay in the network.
In our algorithm, at any moment of time,
the suffix of the execution histories between any two replicas 
differ by at most the maximum number of blocks that 
can be committed during a time period of $O(T \log f')$.
Thus, any inconsistency is limited to only a small period of time.

Musch executes in epochs. An epoch is a slot of time in which $2f'+1$ replicas receive block $B$ proposed by the primary $p$ and agree to commit it. Thus, during each epoch a block is generated and added to the chain. Since $p$ is responsible for aggregating replica signatures for block agreement, if less than $2f'+1$ replica signatures are collected then a view change will be triggered and the primary will be changed. It should be noted that $p$ is also responsible for collecting transactions from clients, ordering the transactions, and sending them to the replicas. 


\subsection{Normal Operation}


As shown in Algorithm \ref{Algorithm2:Primary},
the primary $p$ collects a set of transactions from the clients into an ordered list of transactions $L^a$ (which it will propose in a candidate block) with a sequence number $s$, view number $v$, hash $d=hash(L^a)$, and hash history $h_s = Hash(h_{s-1},d)$ into candidate block
$B= \langle\langle ORDER, s, v, d, h_s \rangle_p , L^a \rangle$.
Primary $p$ then proposes (broadcasts) the candidate block $B$ to each replica $i$.
As shown in Algorithm \ref{Algorithm1:Replica},
upon receipt of $B$ each replica $i$ validates the information,
and then replica $i$ responds to the primary $p$ with the willingness to accept the block in a message $H_i=\langle RESPONSE, s, v,d,i \rangle_i$ to the primary. 

The primary collects at least $2f'+1$ responses from the replicas, aggregates them to $H$, and generates a compressed aggregated signature $\sigma$
\cite{Boneh:2003}. 
Then, the primary broadcasts $\langle  COMMIT, H \rangle_{\sigma}$. 
Upon receipt, each replica $i$ 
verifies $2f'+1$ signatures and the candidate block $B$ commits. 
If verified successfully each replica $i$ responds to the client with the reply message $\langle  REPLY,s, v, c , r ,t, i \rangle_{\sigma_i}$, where $c$ is the client, $t$ is the timestamp and $r$ is the result of execution. Upon receipt of  $f'+1$ valid $REPLY$ messages (which might take $2f'+1$ messages to receive) a client accepts the result.
Assuming a continuous creation of blocks, 
the primary starts the new epoch immediately after the old epoch finishes.
Let $T$ be the maximum delivery delay of a message in the network.
According to the protocol, in the epoch of the new block $B$ with sequence number $s$ there will be two messages that replica $i$ expects to receive from the primary:
(i) the $ORDER$ type message for block $B$ with sequence $s$ within $\Delta_2 = T$ time from the end of the previous epoch, and then 
(ii) the $COMMIT$ type message for block $s$ within $\Delta_3 = 2T$ time since the receipt of the $ORDER$ message.
Therefore, 
the maximum time for an epoch for a replica $i$ is $\Delta_1 = \Delta_2 + \Delta_3$.
A replica $i$ goes into recovery mode 
at time $\Delta_1$ if either of the two expected messages is not received.

\SetKwFor{Upon}{upon}{do}{end}
\SetKwFor{Check}{check always}{that}{end}
\begin{algorithm}[t]

\DontPrintSemicolon 
\caption{Primary $p$}
\label{Algorithm2:Primary}

Latest committed block sequence number is $s_p$\;
\Upon{receipt of transactions from a set of clients $C$}{
Create a block $B$ with sequence number $s_p + 1$\;
Broadcast $B$ to replicas\;
\Upon{receipt of $2f'+1$ hashes $H_i$ of $B$ from replicas}{
Aggregate the hashes into $H$\;
Commit ($B$,$H$)\;
Broadcast $H$ to replicas\;


Send $REPLY$ to client set $C$\;
}
}
\end{algorithm}
\begin{algorithm}[t]
\DontPrintSemicolon
\caption{Replica $i$}
\label{Algorithm1:Replica}

\tcp{Normal Execution}
Latest committed block sequence number is $s_i$\;
\Upon{receipt of block $B$ from primary $p$ with sequence number $s$}{
Calculate hash $H_i$ of block $B$\;
Send $H_i$ to primary $p$\;
\Upon{receipt of aggregated hash $H$ for block $B$ from primary}{
\If{$H$ is signed by at least $2f'+1$ replicas}{Commit $(B,H)$\;
Send $REPLY$ to each client $c$ \;
}
}
}
\BlankLine
\tcp{Special Cases}
\Check{at any time}{
\If{no receipt of expected $s_i+1$ block $B$ or respective hash $H$ within a timeout period}{
Execute Algorithm \ref{Algorithm-3:recovery-i} with parameter $Complain$}
\If{receipt of a block $B$ with sequence $s > s_i + 1$} {
Execute Algorithm \ref{Algorithm-3:recovery-i} with parameter $Complain$}
\If{receipt of valid set of complains $S$ with $f'+1$ complainers}{
Execute Algorithm \ref{ViewChange-Replica} \tcp{initiate view change}
}
}
\end{algorithm}


\subsection{Recovery Mode}
\label{Recovery mode}

In BFT protocols, when a replica detects an error it broadcasts complaints 
to all replicas in the network.  In contrast to this, a replica in Musch during a failure event will only complain to a subset of replicas in the network called window nodes. If $i$ did not receive a response from the current window then the replica complains to the next window of double size until it receive response from at least one correct replica. 
The window sequences are fixed, $W_1, W_2, \ldots, W_{k'}$,
where $k' = \lceil \lg(f' + 1) \rceil$.
Suppose the replica IDs are taken from the set $\{1, \ldots, n\}$ 
and sorted in ascending order (see Fig.~\ref{Figure:window}).
The window $W_1$ consists of a single node with the smallest ID,
Window $W_2$ consists of two replicas with the next IDs in order, 
Window $W_3$ consists of four replicas with the next higher IDs, and so on.
Therefore the window $W_j$ consists of $2^{j-1}$ replicas,
whose IDs are ranked between $2^{j-1}, \ldots, 2^{j}-1$.
During the execution of the algorithm,
the maximum window that will be contacted 
is actually $k = \lceil \lg(f + 1) \rceil$, where $k \leq k'$, 
since this guarantees that at least one correct node will be encountered
among all the window nodes from $W_1$ up to $W_k$.

\SetKwInput{Param}{Parameters}

\begin{algorithm}[t]
\DontPrintSemicolon
\caption{Fault Recovery in Replica $i$}
\label{Algorithm-3:recovery-i}


\Param{$Complain$ from $i$}

Let $l$ be the block sequence number in $Complain$ for which $i$ has not received either $B$ or $H$\;

$j=1$ \tcp{window index}


\tcp{current window is $W_j$} 

\eIf{$i \in W_j$}{
\label{loopline}
\tcp{all window nodes prior to $W_j$ are faulty}
Broadcast $COMPLAIN$ message to replicas\;} 
{
\tcp{$i$ is in a later window than $W_j$ or $i$ is not a window node at all}
Send $COMPLAIN$ to all nodes in window $W_j$\;
\If{there is no commit by a certain timeout}
{ $j = j + 1$ \tcp{increase window}
Goto Line \ref{loopline} 
}

}
\tcp{listen for responses}

Let $l' \geq l$ be the expected sequence number of blocks in the time period since $Complain$ issued

\Upon{receipt of blocks and respective hashes up to at least $l'$}{

Commit all received pairs of block and hash $(B,H)$
}
\end{algorithm}


\begin{algorithm}[t]
\DontPrintSemicolon  
\caption{Window Node $i$}
\label{Window Nodes}

\Upon{receipt of $COMPLAIN$ or $PROOF$ message from replica $j$}{
	\If{$COMPLAIN$ from $j$ is valid}{

		Add $COMPLAIN$ by distinct complainer to the set of complains $S$\;

        \eIf{distinct number of complainers in $S$ is at least $f'+1$}{
        	Broadcast $S$\;
        	Execute Algorithm \ref{ViewChange-Replica}\;
            Reset $S$ to empty\;}   
        {Let $l$ be the sequence number of block requested in $COMPLAIN$\;
         \If{$i$ has the $l$th block and its hash}{
          	Send all blocks and respective hashes starting from sequence $l$ up to the latest to replica $j$\;
          }
        }
   }
   \ElseIf{$PROOF$ is valid}{
     Broadcast $PROOF$ to replicas\;
     Execute Algorithm \ref{ViewChange-Replica}\;
     Reset $S$ to empty\;
   }
}
\end{algorithm}

Algorithm \ref{Algorithm-3:recovery-i} describes how a replica complains to the window(s),
and Algorithm \ref{Window Nodes} shows the respective reactions from the window nodes.
As shown in Algorithm \ref{Algorithm-3:recovery-i},
if replica $i$ complains that it didn't receive expected message ($ORDER$ or $COMMIT$) from $p$ during normal operation, it sends the complaint in the form of $\langle COMPLAIN, s, v, d \rangle_i$, where $d$ and $s$ belongs to the last committed block in the chain of $i$. If it complains to a window $W_j$,
this message is sent to all nodes in $W_j$ which will then know that replica $i$ does not have $ORDER$ or $COMMIT$ messages after block $s$. If replica $i$ has received a message from the primary that proves the maliciousness of $p$, then it attaches the proof in its complaint  $\langle COMPLAIN, PROOF \rangle_i$ to $W_j$.

When $i$ enters the recovery mode
it first complains to window $W_1$, which has a single node.
If $i$ doesn't get any useful response from $W_1$ then it complains to $W_2$,
which has two nodes, so it informs both nodes.
This process can repeat until $i$ contacts all nodes in $W_k$, the last window.
It is guaranteed that replica $i$ will get a response from a correct node in one of these windows.
As shown in Algorithm~\ref{Window Nodes},
the window nodes respond to complaints by returning the requested information.
If they do not have it then they call themselves Algorithm \ref{Algorithm-3:recovery-i} as well.
If the complainer $i$ is a window node itself, it will stop until it reaches its own window size and will broadcast the complaint. Upon broadcast it is guaranteed that it will receive response. The response can be either receipt of missing messages or a view change. If replica $i$ received the missing messages it will forward it to the complainers that it knows, else it will result in view change (primary will be replaced).

Note that regular replicas and window nodes may be complaining at the same time
and probably for the same reason.
A regular node will have to wait for the window nodes to first obtain a response.
It is important to coordinate the actions of the windows nodes
and the regular replicas to receive the responses efficiently without message replication.
For a regular replica $i$,
the timeout period for waiting a response from the window $W_j$ is at most $\Lambda_j = j3T + 6T$. As it takes $\Delta_1=3T$ to detect timeout for the current epoch, then it takes at most $j3T$ to receive a message from the previous window and send the message back to the replica. In case window $j$ does not receive a message from window $j-1$, it will broadcast its complaint and it is guaranteed that it will receive a response, which it will send back to replica $i$ ($3T$).  
From the start time $t$ of the current epoch,
if $i$ does not get a response within $t + \Lambda_j$ then it will contact the next window $W_{j+1}$.


\subsection{View Change}

\label{View Change}

A view change can be triggered if a correct window node $i \in W_l$ receives at least $f'+1$ distinct replica complaints (against  primary $p$) as shown in Algorithm \ref{Window Nodes}. This guarantees that at least one of the complaints is coming from a correct replica. 

Another reason for view change can be the receipt of an explicit $PROOF$ against $p$ by window node $i$.
Once view change is triggered, window node $i$ broadcasts the set of $COMPLAIN$ or $PROOF$ messages it has received to all replicas (Algorithm \ref{Window Nodes}).

Without loss of generality, consider the case where window node $i$ has sent $PROOF$ to all replicas (the same mechanism also applies to other sets of $COMPLAIN$ messages). 
Upon receipt of $PROOF$ 
a replica $j$ increments its view number ($v=v+1$) and assigns new primary $p'$ (namely, $p' =  v \mod n$) (Algorithm \ref{ViewChange-Replica}). 
 
Replica $j$ then adds its most recent block hash $d$ and block number $s$ in the message along with $PROOF$ in a message $\langle VIEWCHANGE, PROOF , s,d,j \rangle_j$ and sends it to the new primary $p'$ (Algorithm~\ref{ViewChange-Replica}). 

Upon receipt of at least $2f'+1$ view change messages from different replicas, $p'$ stores them into set $Q$. 
Then, $p'$ broadcasts $\langle Q \rangle_{\sigma}$, where $\sigma$ is an aggregated signature for all replicas involved in $Q$ (Algorithm \ref{ViewChange-Primary}). Upon receipt of this message, each replica recovers the latest block  history. Assume $s'$ is the highest block number committed so far in the chain.
The block $s'$ must have been committed by at least $2f'+1$ replicas,
and since $Q$ has size at least $2f'+1$,
it must be that $f'+1$ replicas in $Q$ have also committed $s'$,
one of which is a correct node.
Thus, every replica upon receipt of $Q$ can figure out that the latest committed valid block number is $s'$.

Once $s'$ is known, a replica $i$ will check if block with sequence  $s'$ is the latest block in its history $h_i$,
and if it is, $i$ sends a confirmation message $s'_i$ to $p'$ (Algorithm \ref{ViewChange-Replica}).


In this case, at least $f'+1$ correct replicas know the latest block of $p'$ ($s'_{p'}$). If $ s'$  is same as $s'_{p'}$ then $p'$ begins updating all other replicas that have fallen behind (Algorithm ~\ref{ViewChange-Primary}). $p'$ will not send any block generated earlier than the water mark $H$ (Section \ref{section:Checkpoints}).
If $p'$ does not have $s'$ as its latest block then at least $f'+1$ correct replicas know about it and they send missing blocks and their respective $COMMIT$ messages to $p'$ and then $p'$ updates other replicas as described above. Once $p'$ has updated other replicas it will wait to receive at least $2f'+1$ correct replicas have sent confirmation $\  s'_i $ (Algorithm \ref{ViewChange-Primary}).


 
Since there are at least $2f'+1$ correct replicas, $p'$ signs the latest block in their histories that $p'$ has received using an aggregated signature $V \gets \bigcup_i \langle s'_i \rangle_{\sigma}$ and broadcasts it to the replicas. Upon receipt of $V$ each replica is now ready for the new epoch of the next block and is waiting to receive an $ORDER$ message from the new primary $p'$ (Algorithm \ref{ViewChange-Replica}).
In case a replica does not receive expected messages ($Q$, $V$ or blocks and their hashes within a certain expected time),
then it issues a new complaint which is processed similar to the other types of complaints as described above.

During the view change process  there may be some clients who send their request but it will not be processed because replicas are busy. To address this as we mentioned earlier the client $c$ will broadcast its request after epoch time $\Delta_1$, if it did not receive the response from $p$. In such case, all replicas receive the request $T_c$ and forward it to the $p$. Upon receipt of $2f'+1$ such forwarded requests,  $p$ considers $T_c$ to be included in the  $ORDER$ message as soon as possible. $p$ will have to propose those backlogged requests before proposing the new requests it receives. If it proposes a request that has not been seen by  $2f'+1$ replicas (of which $f'+1$ replicas are correct/honest replicas) proposing the backlogged transactions then the replicas can send a complaint which will result in a view change.

\begin{algorithm}[t]
\DontPrintSemicolon
%
%
%
\caption{Replica $i$ View Change}
\label{ViewChange-Replica}
Select new primary $p'=v \mod n$\;
Send $VIEWCHANGE$ containing latest local block number $s_i$ to $p'$\;
Receive aggregated $VIEWCHANGE$s $Q$ from $p'$\;
\If {$Q$ contains at least $2f'+1$ $VIEWCHANGE$s}{
Get the latest block number ($s'$)
that has been signed by at least $f'+1$ replicas in $Q$\;
\eIf {latest block $s_i$ in replica $i$ is same as $s'$}
{ Replica has not lost any block}
{
Receive messages (blocks and their respective hashes) up to $s'$ from $p'$ before timeout
}
Once updated ($s'_i = s'$) send $s'_i$ to $p'$ \;
Receive $V$ from $p'$ containing aggregated histories of at least $2f'+1$ replicas\;
}
\end{algorithm}

\begin{algorithm}[!htp]
\DontPrintSemicolon
\caption{New Primary View Change}\label{ViewChange-Primary}
Receive $VIEWCHANGE$ messages from replicas\;
Aggregate at least $2f'+1$ $VIEWCHANGE$ messages into $Q$\;
Broadcast $Q$ to replicas\;
Get the latest block number ($s'$) that has been signed by at least $f'+1$ replicas in $Q$\;

\eIf {latest block in $p'$ is same as $s'$} {New primary $p'$ has not lost any block}
{
Receive messages (blocks and their respective hashes) up to $s'$ from $f'+1$ replicas that are up to date 
}
Send messages with missing blocks and hashes to all replicas $i$ who have fallen behind, $s_i < s'$, where $s_i$ should not be less than latest water mark\;
Once received updated $s'_i$ from each replica $i$, where $s'_i = s'$, aggregate $s'_i$ into $V$\;
Broadcast $V$\;
\end{algorithm}

\section{Checkpoints}

\label{section:Checkpoints}
As an optimization to the protocol, we  use checkpoints to improve on the number of messages exchanged during view change.
Checkpoints are typically used as a way to truncate the log in other BFT-based protocols \cite{Castro:1999:PBF:296806.296824}. 
In addition to that, we can also use it to prevent malicious replicas from downloading older messages from a new primary $p'$ and delaying the completion of the view change process. As we know from Section \ref{Recovery mode}, some correct replicas might miss messages and go into recovery mode. These replicas need to download those missing messages. But  malicious replicas might try to download very old blocks and delay the view change process. To bound this we use checkpoints. To maintain the safety condition it is required that at least $2f'+1$ replicas agree on the checkpoint. The checkpoint is created after a constant number of blocks (e.g., sequence number divisible by 200). In Musch, replicas can agree on checkpoints during block agreement (checkpoint number to be added to the $RESPONSE$ message). A checkpoint that is agreed upon by $2f'+1$ replicas of which at least $f'+1$ are honest is called a stable checkpoint. Checkpoints have low and high watermarks. Low watermark $h$ is the last stable checkpoint and the high water mark $H$ is the sum of low water mark and $k$ number of blocks($H=k+h$),  where $k$ is large enough (i.e. $k=400$).
If a  replica wants to download a block older than $H$, $p'$ will ignore the download request and might think that the replica is maliciously trying to delay the view change process.

\section{Correctness Analysis}
\label{section:correctness}

In this section we provide proof of correctness and analysis of the Musch protocol. Before we proceed, it is important to define transaction completion and protocol correctness for the Musch protocol.
We say that a transaction $T_c$ issued by a client $c$ is considered to be completed by $c$ if $c$ receives at least $f'+1$ valid $\langle  REPLY,s, v, c , r ,t, i \rangle_{\sigma_i}$ messages.  It is guaranteed that upon receipt of $2f'+1$ $REPLY$ messages from different replicas at least $f'+1$ of them are valid.
We will prove that Musch satisfies the following correctness criteria:

\begin{definition}[Liveness] 
Every transaction proposed by the correct client will eventually be completed in finite time.
\end{definition}

\begin{definition}[Safety] 
A system is safe if a correct primary proposes a block of ordered transactions with block number $s$ and it is committed by at least $2f'+1$ replicas, then any block that has been committed earlier will  have smaller block number ($s' < s$) in the chain. Thus, block $B_{s'}$ will be the prefix of block $B_s$ in the chain.
Additionally the order of transactions within the block will remain identical in all correct replicas (due to Merkle tree\footnote{Merkle trees are hash-based data structures in which each leaf node is hash of a data block and each non leaf node is hash of its children. It is mainly used for efficient data verification.}).
\end{definition}

\subsection{Safety}
\begin{lemma}
\label{Lemma1}
Any two committed blocks $B_{s'}$ and $B_{s}$ must have a different block number.
\end{lemma}
\begin{proof}
Consider committed blocks $B_{s'}$ and $B_s$.
At least a set of $2f'+1$ replicas $S_1$ have agreed to all transactions with $B_{s'}$ and have committed it. Similarly, at least a set of $2f'+1$ replicas have agreed for the transactions in block $B_{s}$ and committed it. 
Since there are $3f'+1$ replicas, there is at least one correct replica (out of the at least $f'+1$ replicas in $S_1 \cap S_2$) that committed both for $B_{s'}$ and $B_{s}$. But a correct replica only commits one block with a specific block number. Thus, both blocks must have different numbers. The same mechanism applies during recovery mode.
\end{proof}

\begin{lemma}
\label{Lemma2}
If block $B_{s'}$ commits earlier than block $B_{s}$, then $B_{s'}$ has a smaller block number than $B_{s}$. 
\end{lemma}
\begin{proof}
As per Lemma \ref{Lemma1}, at least one correct replica $k$ has committed both $B_{s'}$ and $B_{s}$.
Suppose, that $B_{s'}$ gets a block number $s'$ 
which is smaller than the block number $s$ of $B_s$, that is $s'< s$ ($s \neq s'$ from Lemma \ref{Lemma1}).
A correct Replica $k$ will only accept $B_s$ if $B_s$ is consistent with its local history (only if $s>s'$ ).
\end{proof}


\begin{lemma}
\label{Lemma3}
Musch is safe during view change.
\end{lemma}
\begin{proof}
During a view change (Algorithms \ref{ViewChange-Replica} and \ref{ViewChange-Primary}), all replicas including the new primary $p'$ retrieve the latest history and block number $s'$ 
as at least $f'+ 1$ replicas will agree on the latest block number $s'$, which includes a correct replica that knows $s'$. 
All correct replicas know the latest block $s'_{p'}$ in the history of $p'$  from $\langle Q\rangle_{\sigma}$. If $s'_{p'}=s'$  then $p'$ begins updating all other replicas that have fallen behind in history, in other words it updates all the replicas that do not have blocks and respective $COMMIT$ messages up to $s'$. 
If $p'$ does not have $s'$ as its latest block then at least $f'+1$ correct replicas know about it (from $Q$) and they send missing blocks and their respective $COMMIT$ messages to $p'$ and then $p'$ updates other replicas as described above. Once $p'$ updated (receive blocks and $COMMIT$s up to $s'$) other replicas it will take $T$ timeout period to receive $s_i $ (update confirmation) from at least $2f'+1$ replicas. Then, $p'$ signs all their histories using an aggregated signature $V \gets \bigcup_i \langle s_i \rangle_{\sigma}$ and broadcasts it to the replicas. Upon receipt of $V$ each replica is now ready for the new epoch of the next block and is waiting to receive an $ORDER$ message from the new primary $p'$.
%
%
\end{proof}

\begin{theorem}[Safety]
Musch is safe.
\end{theorem}
\begin{proof}
Lemma \ref{Lemma3} guarantees safety when the new primary $p'$ is correct.
If $p'$ is not correct, safety will be guaranteed when eventually a correct 
primary will be chosen.
Therefore, based on Lemmas \ref{Lemma1}, \ref{Lemma2} and \ref{Lemma3}, Musch is safe when replicas are either in normal, recovery, or view change mode.
\end{proof}

\subsection{Liveness}
In this section we provide a proof for liveness of Musch.
\begin{lemma} \label{Normal Liveness}
Musch satisfies liveness when the primary is correct.
\end{lemma}

\begin{proof}
Consider a correct primary $p$ that executes Algorithm \ref{Algorithm2:Primary},
and also the replicas that execute Algorithm \ref{Algorithm1:Replica}.
Primary $p$ receives at least $2f'+1$ correct $RESPONSE$ messages from replicas,
aggregates and signs them using an aggregation signature $\sigma$. It then broadcasts the signed $COMMIT$ message to all replicas. Upon receipt of the $COMMIT$ message each replica will commit the block. The primary $p$ along with all correct replicas also forwards a reply message to each client $\langle  REPLY,s, v, c , r ,t, i \rangle$ and clients will mark the transaction as completed. 
\end{proof}



%
%

\begin{lemma}\label{f'+1 complains}
If there are $f'+1$ complaints, or there is a complaint with a proof of maliciousness against the primary, then a view change will occur.
\end{lemma}
\begin{proof}
Algorithm \ref{Algorithm-3:recovery-i} guarantees that, in the worst case, a replica $i$ can find  a window node $W_{k}$ to complain, where,  $k = \lceil \lg(f + 1) \rceil$ and $W_{k}$ contains at least one correct replica,
since $W_k$ contains at least $2^{\lg(f+1)} = f+1$ nodes.
Observe that once a replica $i$ has found a honest window node, it is guaranteed that the honest node will reply to its valid complaint either by sending back blocks and $COMMIT$s or if the number of complaints are greater than $f'$, then the window node will broadcast all complaints to the network causing a view change (Algorithm \ref{Window Nodes}). 

If a replica $j \in W_{k}$ receives at least $f'+1$ complaints from other replicas it triggers a view change according to the Algorithm \ref{Window Nodes}. Since $f'+1$ complaints are received, this guarantees that at least one honest replica has complained.

Similarly, $j$ may receive an explicit proof that the primary $p$ is faulty
($p$'s history is incorrect, or it has proposed an invalid transaction, etc.). In such a case only one complaint is needed to prove that $p$ is malicious and a view change will be triggered.
\end{proof}
 

\begin{lemma} \label{liveness malicious primary}
If a transaction is not completed then a view change will occur.
\end{lemma}

\begin{proof}
If a transaction does not complete after sufficient time $\Delta_1$, then the client $c$ broadcasts its transaction $T_c$ to the replicas. Upon receipt of $T_c$, the replicas check if they have already committed a block that contains $T_c$. If they did, each replica $i$ will send $\langle ACK, T_c\rangle_i$ to the client and upon receipt of $2f'+1$ $ACK$ messages the client will consider the transaction as complete. If primary $p$ has not proposed the transaction $T_c$, then each replica will forward $T_c$ to $p$ and will expect that $p$ will include it in the next $ORDER$ message (during normal operation). If $p$ does not include it in the next $ORDER$ message, then replicas will start complaining, which will result in a view change (if at least $f'+1$ replicas complain, from Lemma \ref{f'+1 complains}).

Another case that can prevent a request from being committed is when  
replicas receive a $COMMIT$ message signed by less than $2f'+1$ replicas. In this case, this can be used as proof against $p$ and a complaint can be made, which will result in a view change (Lemma \ref{f'+1 complains}).
\end{proof}

\begin{lemma} \label{liveness view change}
Musch satisfies liveness even if a client request is received during a view change.
\end{lemma}

\begin{proof}
During the view change process, there may be some clients who send their request for transaction $T_c$ but it will not be processed because replicas are busy with the view change. To address this, as mentioned earlier the client $c$ will broadcast its request after epoch timeout $\Delta_1$, if it did not receive a response from $p$. In such a case, all replicas receive the request $T_c$ and forward it to the new primary $p'$. Upon receipt of $2f'+1$ such forwarded requests the $p'$ considers $T_c$ to be included in the $ORDER$ message as soon as possible. The new primary $p'$ will have to propose those backlogged client requests during the view change, before proposing the new requests it receives. If it proposes a request that has not been seen by $2f'+1$ replicas (of which $f'+1$ replicas are correct/honest replicas), proposing the backlogged transactions then the replicas can start complaints, which will result in a new view change (Lemma \ref{f'+1 complains}).
\end{proof}

\begin{theorem}[Liveness]
\label{thrm:liveness}
Musch satisfies liveness and all correct transactions will be completed eventually.
\end{theorem}
\begin{proof}
Based on Lemmas \ref{Normal Liveness}, \ref{liveness malicious primary} and \ref{liveness view change}, any correct transaction request by a client will be completed within a finite period of time.
\end{proof}

\section{Communication Complexity}
\label{section:complexity}

In communication complexity, we count all messages that cause a reaction in our algorithm and we refer to these as {\em effective messages}.
In contrast, there are {\em ineffective} messages,
which have sources that have been identified as malicious,
and so the recipient can ignore these messages.
We will measure the number of effective messages exchanged in an epoch,
and we will consider worst cases scenarios, with or without view change.
In other words, we consider worst-case {\em performance attacks} when malicious replicas attempt to increase the
communication of the protocol by causing messages to be sent from correct replicas.

In the communication complexity we consider separately the messages sent between clients and replicas, and those sent only between replicas.

\subsection{Client-Replica Communication Complexity} \label{Replica-Client complexity}
If a client sends a transaction to the primary $p$,
and does not receive a response from the primary $p$ within $\Delta_1$, 
then the client broadcasts to the primary $p$ (a broadcast involves $n$ messages). 
Upon receipt of a broadcast from a client, if replica $i$ has already processed the client's transaction it will answer to the client with an acknowledgement.
If not, the replica $i$ will forward the client's request to the primary, forcing it to process it as soon as possible.
The liveness property of our algorithm, Theorem \ref{thrm:liveness},
will guarantee that eventually at least $2f'+1 = O(n)$ of the replicas will 
send acknowledgements to the client.
Therefore, we get the following result:

\begin{lemma}
\label{lemma:client-messages}
For each transaction sent by a client, at most $O(n)$ messages will 
be exchanged between the client and the replicas in order to process the transaction (i.e., include the transaction in a block).
\end{lemma}



\subsection{Replica-Replica Communication Complexity}
In this section we analyze the communication complexity of the consensus engine of our protocol, which includes the primary $p$ and the replicas (in total $n$ nodes).
A malicious primary $p$ and malicious replicas 
both can try to increase the communication complexity.

\subsubsection{Messages caused by malicious primary}
Let $R_c$ be the set of replicas that complain.
First, we examine the case when the nodes in $R_c$ did not receive the block or $COMMIT$ message and they complain. A malicious primary $p$ can afford not to send such messages up to at most $f'$ replicas, without getting caught as being malicious; that is, $|R_c| \leq f'$.

In this case, each of the complainers in $R_c$ may have to communicate with up to $2f+1$ window nodes, since this guarantees a window that has at least one correct window node. This gives at most $(2f+1)|R_c|$ messages.
In the worst case, out of the $2f+1$ window replicas at most $f+1$ will be the honest ones that will broadcast to all $n$ replicas and will receive their response, to be forwarded to the complainers $R_c$, giving at most $2(f+1)n + (f+1)|R_c|$ additional messages. The total communication complexity in this case will be (since $|R_c| \leq f' < n/3$):

\begin{equation}\label{eqn:primary-msg}
\begin{split}
(2f+1)|R_c| +2 (f+1) n + (f+1) |R_c| \\ \leq (5f + 4)n = O(fn+n). 
\end{split}
\end{equation}

\subsubsection{Messages caused by malicious replicas}
Suppose the set of complainers $R_c$ are malicious, 
thus, $|R_c| \leq f$.
Window nodes do not respond to repetitive complains from the same replica
(non-effective messages), which prevents malicious replicas from increasing the communication complexity.
Nevertheless,
each window node may respond once to each malicious request.
A window node $j$ can respond to a complain message in the following ways:
\begin{itemize}
\item 
If window node $j$ has the appropriate response to the complain (i.e. it has the block or $COMMIT$) it will send it back to the replica that complained.
At most $2f'+1$ window nodes will be accessed by each replica in $R_c$, since this is the bound on the total number of window nodes.
Therefore, in this case, the number of messages are at most:
\begin{equation}\label{eqn:replica-msg1}
\begin{split}
2 (2f'+1)|R_c| \leq (4f'+2)f \\ < (4n/3 + 2)f = O(fn+n).
\end{split}
\end{equation}

\item
If window node $j$ does not have the appropriate response (block or $COMMIT$),
then $j$ itself is also executing the window protocol from smaller to larger windows,
and when it eventually points to its own window, it will broadcast the complaint to get a response from other replicas (acting as a regular window node).  This scenario can only happen if all the previous windows are populated by $f$ faulty nodes. 
The number of complaints from $R_c$ to up to $2f'+1$ window nodes 
are bounded by $(2f'+1)|R_c|$.
Similarly,
the respective responses are bounded by $(2f'+1)|R_c|$.
For calculating the messages from the broadcasts,
out of the $2f'+1$ total window nodes,
at most $2f+1$ window nodes will react to the received complaints with broadcasts,
since the first encountered window of size at least $f+1$ will respond to any complaint from a valid node.
Thus, each of the up to $2f+1$ windows nodes broadcasts to all replicas,
causing $(2f+1)n$ additional messages.
Therefore, in this case, 
the number of messages are at most:
\begin{equation}\label{eqn:replica-msg2}
\begin{split}
2(2f'+1)|R_c|+ (2f+1)n \\ < (4n/3 + 2) f + (2f+1)n = O(fn+n).
\end{split}
\end{equation}
\end{itemize}

\subsection{View Change Communication Complexity}
When a correct window node receives $f'+1$ complaints it will broadcast all of them to all replicas ($n$ messages). 
There are at most $f+1$ window nodes that will broadcast (since those window nodes could be correct in the last accessed window size),
resulting to at most $(f+1)n$ messages.
Upon receipt of the broadcast message,
each replica begins the view change process. The replica sends back a $VIEWCHANGE$ message to the new primary $p'$ which also includes its history ($n$ messages). The new primary $p'$ aggregates all $VIEWCHANGE$ messages into $\langle Q \rangle_{\sigma_p'}$ and broadcasts ($n$ messages). Upon receipt each replica extracts the most recent block as described in Section \ref{View Change}.
Therefore, the number of messages from this part of the algorithm is at most: 
\begin{equation}\label{eqn:viewchange-msg1}
(f+1)n+n+n= fn + 3n.
\end{equation}
During this, at least $f'+1$ correct replicas have the latest committed block $s'$, and this block is chosen as the starting point for the next epoch, which will build another block ($s'+1$) over it. All other replicas that have block number less than $s'$ as their latest block have to download all the blocks up to $s'$ from $p'$. If $p'$ does not have $s'$ as its latest block, then $f'+1$ replicas that have it will bring $p'$ up to date.
Thus, if $f'+1$ replicas have $s'$ as their latest block, then, at most $2f'$ replicas in the worst case get (download) messages up from the high water mark in checkpoint $H$ to $s'$. Let $e$ be the number of committed blocks from $H$ to $s'$. For each committed block we need two messages, first the block itself and the second is the $COMMIT$ message. 
Thus, we have:
\begin{equation}\label{eqn:viewchange-msg2}
2e (2f'+f') = 6ef'.
\end{equation}
Assuming frequent checkpoints (say every a fixed number of blocks), 
we can assume that $e$ is a constant.
From Equations \ref{eqn:viewchange-msg1} and \ref{eqn:viewchange-msg2} we have
for the total number of messages in view change:
\begin{equation}\label{eqn:viewchange-msg-final}
fn + 3n + 6e(n/3) = O(fn+n).
\end{equation}

\subsection{Overall Messages}
Combining Equations \ref{eqn:primary-msg}, \ref{eqn:replica-msg1}, \ref{eqn:replica-msg2}, we obtain $O(fn+n)$ communication complexity in a single epoch
for the communication complexity between replicas.
From Equation \ref{eqn:viewchange-msg-final} 
the communication complexity is also $O(fn+n)$ during view change.
Therefore, we have the following result:

\begin{lemma}
\label{lemma:replica-messages}
The number of messages exchanged 
between replicas in an epoch or during view change are $O(fn+n)$.
\end{lemma}

Combining Lemmas \ref{lemma:client-messages} and \ref{lemma:replica-messages}
we obtain the main result for the communication complexity:

\begin{theorem}[Communication complexity]
For $\tau$ initiated transactions in an epoch, 
the communication complexity is $O(\tau n + fn)$.
For constant $\tau$, the communication complexity is $O(fn+n)$.
\end{theorem}

\section{Conclusions}
\label{section:conclusion}
In this paper we proposed Musch, 
a BFT-based consensus protocol,
in an effort to avoid excessive messages and improve the scalability of blockchain algorithms.
Through the use of windows, the algorithm adapts to the actual number of faulty nodes $f$, and in this way 
it avoids unnecessary messages.
This improvement does not sacrifice 
on the latency, since our algorithm still uses a small number of communication rounds. For future work, it would be interesting to investigate whether we can decrease the message complexity further, i.e. to $O(n)$ under $f$ faults, by introducing an intelligent scheme to detect faulty nodes and foil attempts to increase message complexity.
\bibliographystyle{IEEEtran}
\bibliography{OPBFT}
\end{document}